\documentclass[a4paper]{article}
\usepackage{latexsym}
\usepackage{amsmath,amssymb,amsfonts}

\newtheorem{prethm}{{\bf Theorem}}

\newenvironment{thm}{\begin{prethm}{\hspace{-0.5
               em}{\bf.}}}{\end{prethm}}

\newtheorem{prepro}[prethm]{Proposition}

\newtheorem{prelem}[prethm]{Lemma}
\newenvironment{lem}{\begin{prelem}{\hspace{-0.5
               em}{\bf.}}}{\end{prelem}}

\newtheorem{precor}[prethm]{Corollary}

\newtheorem{preremark}{{\bf Remark}}

\newtheorem{preexample}{{\bf Example}}

\newtheorem{preproof}{{\bf Proof.}}

\newenvironment{proof}[1]{\begin{preproof}{\rm
               #1}\hfill{$\Box$}}{\end{preproof}}

\input{epsf}

\title{\bf\ On z-factorization and c-factorization of standard episturmian words}

\author{{\normalsize{ M. Mohammad-Noori${}^{ \textrm{a}, \,1}$}, { N. Ghareghani${}^{ \textrm{b, c}}$}, {  P. Sharifani ${}^{ \textrm{a}}$}\,
}\vspace{2mm} \\{\footnotesize{$^{ \textrm{a}}$\it
Department of
Mathematics, Statistics and Computer
Science, University of Tehran,}}\vspace{-2mm}\\{\footnotesize{\it   Tehran,
Iran}}\\
{\footnotesize{$^{ \textrm{b}}$\it Department of Mathematical Sciences, K.N. Toosi University of Technology}}\vspace{-2mm}\\
{\footnotesize{$^{ \textrm{c}}$\it School of Mathematics, Institute for Research in
Fundamental Sciences {\rm(IPM),}}}\vspace{-2mm}\\{\footnotesize{\it P.O.Box: 19395-5746, Tehran,
Iran}}
\\{\footnotesize Emails: morteza@ipm.ir, mnoori@khayam.ut.ac.ir},\\
 {\footnotesize ghareghani@ipm.ir, Psharifani@khayam.ut.ac.ir}
 }

\begin{document}

\maketitle \footnotetext[1]{\tt Corresponding author}


\begin{abstract}
\noindent Ziv-Lempel and Crochemore factorization are two kinds of factorizations of words
related to text processing. In this paper, we find these factorizations for standard epiesturmian words.
Thus the previously known c-factorization of standard Sturmian words is provided as a special case.
Moreover, the two factorizations are compared.
\end{abstract}
\vspace{3mm}

\section{Introduction}
Some factorizations of finite words are studied by Ziv and Lemplel in a seminal paper \cite{ZivLempel}. These factorizations are related to information theory and text processing. Several years later, Crochemore introduced another factorization of words for the design of a linear time algorithm to detect squares in a word \cite{croch, crochHan, crochRyt}. The Ziv-Lempel and Crochemore factorizations seem to be similar in some cases but significantly different in some other examples.

 In \cite{Bersavelli}, Crochemore factorizations of some of well-known infinite words, namely characteristic
 Sturmian words and (generalized) Thue-Morse words and the period doubling sequence, are explicitly given based
 on their combinatorial structures. Also, they have shown that in general, the number of factors in the Crochemore
 factorization is at most twice the number of factors of the Ziv-Lempel factorization.\\

The Crochemore factorization (or c-factorization for short) of a
word ${\bf w}$ is defined as follows: Each factor of $c({\bf w})$ is
either a fresh letter, or it is a maximal factor of ${\bf w}$, which
has already occurred in the prefix of the word. More formally, the
c-factorization $c({\bf w})$ of a word ${\bf w}$ is
$$c({\bf w})=(c_1,\cdots,c_{m},c_{m+1},\cdots),$$
where $c_m$ is the longest prefix of $c_m c_{m+1}\cdots$ occurring twice in $ c_1 \cdots c_m$, or $c_m$ is
a letter $a$ which has not occurred in $c_1\cdots c_{m-1}$.\\

The Ziv-Lempel factorization (or z-factorization for short) of a word ${\bf w}$ is
$$z({\bf w})=(z_1,\cdots,z_{m},z_{m+1},\cdots),$$
where $z_m$ is the shortest prefix of $z_m z_{m+1}\cdots$ which
occurs only once in the word $z_1\cdots z_m$.
In this paper, we give explicit formulas for z-factorization and c-factorization of standard episturmian words
, thus we obtain the previous c-factorization of standard Sturmian words in \cite{Bersavelli} as a special case. Moreover, these results reveal the relation between two factorizations in the case of standard episturmian words.
The rest of the paper is organized as follows. In Section 2 we present some useful definitions and notation of combinatorics on words. Section 3 is devoted to review the definition and some properties of episturmian words.
In Section 4, we study z-factorization of standard episturmian words. Finally in Section 5 we present a result about the c-factorization of standard episturmian words.

\section{Definitions and notation}
We denote the alphabet (which is finite) by ${\mathcal A}$. As usual, we denote by ${\mathcal A}^*$, the set of words over ${\mathcal A}$ and by $\epsilon$ the empty word. We use the notation ${\mathcal A}^+ ={\mathcal A}^* \setminus \{\epsilon\}$.
If $a\in {\mathcal A}$ and $w=w_1w_2\ldots w_n$ is a word over ${\mathcal A}$ with the $w_i\in {\mathcal A}$, then
the symbols $|w|$ and ${|w|}_a$ denote respectively the length $n$ of $w$, and the number of occurrences of letter $a$ in $w$. For an infinite word ${\bf w}$ we denote by $Alph({\bf w})$ (resp. $Ult({\bf w})$) the number of letters which appear (resp. appear infinitely many times) in ${\bf w}$ (The first notation is also used for finite words).
A word $v$ is a factor of a word $w$, written $v \prec w$, if there exists $u,u'\in {\mathcal A}^*$, such that $w=uvu'$. A word $v$ is said to be a prefix (resp. suffix) of a word $w$, written
$v\lhd w$ (resp. $v \rhd w$), if there exists $u\in {\mathcal A}^*$ such that $w=vu$ (resp. $w=uv$). If $w=vu$ (resp. $w=uv$,) we simply write $v=wu^{-1}$ (resp. $v=u^{-1}w$).
The notations of prefix and factor extend naturally to infinite words. Two words $u$ and $v$ are {\it conjugate} if
there exist words $p$ and $q$ such that $u=pq$ and $v=qp$. For a word $w$, the set
$F(w)$ (resp. $F_n(w)$) is the set of its factors (resp. the set of its factors of length $n$); these notations are also used for infinite words. If ${\bf w}$ is an infinite word, then the related complexity function, is $p_{{\bf w}}(n)=|F_n({\bf w})|$. The {\it reversal} of $w=w_1 w_2 \ldots w_{n}$ is $\overline{w}=w_n w_{n-1}\ldots w_1$. The
word $w$ is a {\it palindrome} if $w=\overline{w}$. A word $w\in {\mathcal A}^+$ is called primitive if $m\in \mathbb{N}_{+}$ and $w=u^m$ implies $m=1$.

\section{Episturmian words}
An infinite word ${\bf s}$ is episturmian if $F({\bf s})$ is closed under reversal and for any $\ell \in {\mathbb{N}}$
there exists at most one right special word in $F_{\ell}({\bf s})$. Then Sturmian words are just nonperiodic
episturmian words on a binary alphabet. An episturmian word is {\it standard} if all its left special factors are prefixes of it. It is well-known that if an episturmian word ${\bf t}$ is not periodic and $Ult({\bf t})=k$, then its complexity function is ultimately $p_{\bf t}(n)=(k-1)n+q$ for some $q \in {\mathbb N}_{+}$.
Let ${\bf t}$ be an episturmian word. If ${\bf t}$ is nonperiodic then there exists a unique standard episturmian word ${\bf s}$ satisfying $F_{\bf t}=F_{\bf s}$; If ${\bf t}$ is periodic then we may find several standard episturmian words {\bf s} satisfying $F_{\bf t}=F_{\bf s}$. In any case, there exists at least one standard episturmian word
{\bf s} with $F_{\bf t}=F_{\bf s}$. If the sequence of palindromic prefixes of a standard episturmian word ${\bf s}$ is $u_1=\epsilon,u_2,u_3,\cdots$, then there exists an infinite word $\Delta({\bf s})=x_1 x_2 \cdots$, \, $x_i\in {\mathcal{A}}$ called its {\it directive word} such that for all $n\in {\mathbb{N}}_{+}$,
$$ u_{n+1}=(u_n x_n)^{(+)}$$
where $w^{(+)}$ is defined as the shortest palindrome having $w$ as a prefix.
(Similar construction for Sturmian words can be found in \cite{Luca}.)
The relation between $u_n$ and $u_{n+1}$
 can also be explained using morphisms: For $a\in \mathcal{A}$ define the morphism $\psi_a$ by
 $\psi_a(a)=a$, and $\psi_a(x)=ax$ for $x\in \mathcal{A}\setminus \{a\}$. Let $\mu_0=Id$ and $\mu_n=\psi_{x_1} \psi_{x_2}\cdots \psi_{x_n}$ for $n \in {\mathbb{N}}_{+}$. Moreover, let $h_n=\mu_n(x_{n+1})$. Then
 $$u_{n+1}=h_{n-1}u_n,\, n\in {\mathbb{N}}_{+}$$
 From the above equation, it is concluded that
 \begin{equation} \label{uhhh1} u_{n+1}=h_{n-1}\cdots h_1 h_0=\overline{h_0}\, \overline{h_1} \cdots \overline{h_{n-1}}  \end{equation}

 It is known that for any integer $n$, $h_n$ is primitive (See Proposition 2.8 of \cite{juspiri}) and
 so is $\overline{h_n}$. For any integer $n$ define $P(n)$ as the maximum value of $i$ satisfying $i<n$ and
 $x_i=x_n$; if there is no such $i$ then $P(n)$ is undefined. We have the following Lemma.

 \begin{lem}\label{Pn}
 \begin{enumerate}
 \item[\rm(i)]$$h_{n-1}=\left\{
  \begin{array}{ll}
    {u_n x_n} & \hbox{if $P(n)$ is undefined,} \\
    {u_n u_{P(n)}^{-1}} & \hbox{otherwise.}
  \end{array}
\right.$$

  \item[\rm(ii)] If $P(n)$ is defined then
  $$h_{n-1}=h_{n-2} h_{n-3}\cdots h_{P(n)-1}.$$
\end{enumerate}
\end{lem}
\begin{proof}{
 \begin{enumerate}
 \item[\rm(i)] See the end of Section 2.1 of \cite{juspiri}.
 \item[\rm(ii)] This is proved by using part (i) and (\ref{uhhh1}).
 \end{enumerate}
}
\end{proof}
 It is obvious that $h_{n-1}\lhd h_n$. In addition, by Proposition 2.11 of \cite{juspiri} we have
 \begin{lem}\label{hxPJ}
 \begin{enumerate}
 \item[\rm(i)] $h_n=h_{n-1}$ if and only if $x_{n+1}=x_n$.
 \item[\rm(ii)] If $x_{n+1}\neq x_n$ then $u_n$ is a proper prefix of $h_n$.
 \end{enumerate}
 \end{lem}

\begin{lem}\label{hu}
Let $\Delta ({\bf s})= x_1 \ldots x_n, \ldots,\,$ $x_i\in \mathcal{A}$. Suppose that $x_n = \alpha$ and the letter
$\alpha$ has at least one appearance before $x_n$ in $\Delta (s)$. Then
 \begin{enumerate}
 \item[\rm(i)] $h_{n-1} \lhd u_n$ and $\overline{h_{n-1}}\rhd
 u_n$.
 \item[\rm(ii)] The word $v_{n-1}=u_n (\overline{h_{n-1}})^{-1}$ is palindrome.
  \item[\rm(iii)] $v_{n-1}\rhd u_{n-1}$ and $v_{n-1}\lhd u_{n-1}$.
  \item[\rm(iv)] $u_n\rhd u_{n-1}\overline{h_{n-1}}$.
  \item[\rm(v)] If moreover $x_n\neq x_{n-1}$, then $u_n\rhd (\overline{h_{n-1}})^2$ and $u_{n+1}\rhd (\overline{h_{n-1}})^3$.
 \end{enumerate}
\end{lem}

\begin{proof}{
 \begin{enumerate}
 \item[\rm(i)] By Lemma \ref{Pn}(i), $h_{n-1}=u_n u_{P(n)}^{-1}$. So $h_{n-1} \lhd u_n$, which concludes
 $\overline{h_{n-1}}\rhd
 \overline{u_n}=u_n$.
 \item[\rm(ii)] By part (i), there exists a word $v_{n-1}$ satisfying $u_n=v_{n-1}\overline{h_{n-1}}$. Hence by
 $u_{n+1}=h_{n-1}u_n$, we obtain $u_{n+1}=h_{n-1}v_{n-1}\overline{h_{n-1}}$. But since $u_{n+1}$ is palindromic, from the last equation we conclude that so is $v_{n-1}$.
 \item[\rm(iii)] From $u_n=v_{n-1}\overline{h_{n-1}}=u_{n-1}h_{n-2}$ and $|h_{n-1}|\geq |h_{n-2}|$ we conclude that $v_{n-1}\lhd u_{n-1}$, which yields $v_{n-1}=\overline{v_{n-1}}\rhd \overline{u_{n-1}}=u_{n-1}$.
 \item[\rm(iv)] This is concluded from $u_n=v_{n-1}\overline{h_{n-1}}$ using part (iii).
 \item[\rm(v)] Using part (iii) and Lemma \ref{hxPJ}(ii), we get $u_n\rhd (\overline{h_{n-1}})^2$; combining this with $u_{n+1}=u_n \overline{h_{n-1}}$, we provide $u_{n+1}\rhd (\overline{h_{n-1}})^3$.
 \end{enumerate}
}\end{proof}

The following representation of directive word is useful for next sections. Let
$$ \Delta({\bf s})=x_1 x_2 \cdots= y_1^{d_1} y_2^{d_2} \cdots\, ,$$

where $x_i, y_i \in {\mathcal A}$, $y_i\neq y_{i+1}$ and $d_i>0$ for
$i>0$. Define the function $g:{\mathbb N} \rightarrow {\mathbb N}$
by
$$g(m)=d_1+\cdots+d_{m-1}+1$$

\begin{lem}\label{deltaY}
With the above definitions, the following statements hold.
 \begin{enumerate}
 \item[\rm(i)] $u_{g(m+1)}=(h_{g(m)-1})^{d_m}u_{g(m)}=u_{g(m)} (\overline {h_{g(m)-1}})^{d_m}.$

 \item[\rm(ii)] $u_{g(m+1)}=(h_{g(m)-1})^{d_m} (h_{g(m-1)-1})^{d_{m-1}}\cdots (h_0)^{d_1}= (\overline{h_0})^{d_1} (\overline{h_1})^{d_2}\cdots (\overline{h_{g(m)-1}})^{d_m}.$
  \item[\rm(iii)] $u_{g(m)-1}$ is a proper prefix of $h_{g(m)-1}$.
  \item[\rm(iv)]  $u_{g(m)}\rhd (\overline{h_{g(m)-1}})^2$ and $u_{g(m+1)}\rhd (\overline{h_{g(m)-1}})^{d_m+2}$.
 \end{enumerate}
\end{lem}

\begin{proof}{
 \begin{enumerate}
 \item[\rm(i)] For any integer $n$ with $g(m)\leq n<g(m+1)$ we have $x_n=y_m$ and by Lemma \ref{hxPJ}(i), $h_{n-1}=h_{g(m)-1}$. Thus for any integer $j$ with $0\leq j\leq d_m$ we have
 $$u_{g(m)+j}=(h_{g(m)-1})^j u_{g(m)}=u_{g(m)}(\overline{h_{g(m)-1}})^j.$$
Particularly for $j=d_m$ the result is provided.
\item[\rm(ii)] This is concluded from (\ref{uhhh1}).
\item[\rm(iii)] This is obtained from Lemma \ref{hxPJ}(ii).
\item[\rm(iv)] By Lemma \ref{hu}(v) we obtain $u_{g(m)}\rhd (\overline{h_{g(m)-1}})^2$; Using this and $u_{g(m)+1}=u_{g(m)}(\overline{h_{g(m)-1}})^{d_m}$, we provide $u_{g(m+1)}\rhd (\overline{h_{g(m)-1}})^{d_m+2}$.
\end{enumerate}
}
\end{proof}

\section{z-factorization}

\begin{thm}\label{zFactoEpi} Let ${\bf s}$ be an episturmian word with directive word
$\Delta ({\bf s})=  x_1 x_2 x_3 \ldots = y_1^{d_1} y_2^{d_2}
\ldots$, where $y_{i}\neq y_{i+1}$, for all $i\geq 1$. The
z-factorization of ${\bf s}$ is of the form $z({\bf s})=(z_1, z_2,
\ldots)$, where $z_1=x_1$ and
$z_k=y_{k-1}^{-1}(\overline{h_{g(k-1)-1}})^{d_{k-1}}y_k$ for $k\geq
2$.
\end{thm}
\begin{proof}
{We prove the result by induction on $k$. It is easily seen that
$z_1=x_1=y_1$. Now suppose that the result is true for any $j<k$.
Thus we have
\begin{align*}
z_1 z_2 z_3\cdots z_{k-1}&=y_1 \,\, y_1^{-1}
(\overline{h_{g(1)-1}})^{d_1} y_2 \, \, y_2^{-1}
(\overline{h_{g(2)-1}})^{d_2} y_3\, \cdots \, y_{k-2}^{-1}
(\overline{h_{g(k-2)-1}})^{d_{k-2}}y_{k-1}\\
&= (\overline{h_{g(1)-1}})^{d_1} (\overline{h_{g(2)-1}})^{d_2}
\cdots \overline{h_{g(k-2)-1}})^{d_{k-2}}y_{k-1}\\
&=u_{g(k-1)}y_{k-1}
\end{align*}
We should conclude
$z_k=y_{k-1}^{-1}(\overline{h_{g(k-1)-1}})^{d_{k-1}}y_k$. For this purpose, the two following facts should be proved.\\

{\bf Fact 1.} $y_{k-1}^{-1}(\overline{h_{g(k-1)-1}})^{d_{k-1}}\prec u_{g(k)}x_1^{-1}$.\\

{\bf Fact 2.} $y_{k-1}^{-1}(\overline{h_{g(k-1)-1}})^{d_{k-1}} y_k \nprec u_{g(k)}$.\\

We prove these facts in two cases.\\

{\bf Case (i).} Suppose that $y_{k-1}=\alpha$ has already appeared in $\Delta({\bf s})$.
By Lemma \ref{hu} (i), $\overline{h_{g(k-1)-1}}\rhd u_{g(k-1)}$ hence
$$y_{k-1}^{-1}(\overline{h_{g(k-1)-1}})^{d_{k-1}}\rhd u_{g(k-1)}(\overline{h_{g(k-1)-1}})^{d_{k-1}-1} $$
But the right side, is a prefix of $u_{g(k)}=u_{g(k-1)}(\overline{h_{g(k-1)-1}})^{d_{k-1}}$. This proves Fact 1.\\

To prove Fact(2), by contrary, suppose that
\begin{equation} \label{s1}
y_{k-1}^{-1}(\overline{h_{g(k-1)-1}})^{d_{k-1}} y_k \prec u_{g(k)}.\end{equation}

\noindent By Lemma \ref{deltaY}(iv), $u_{g(k-1)}\rhd (\overline{h_{g(k-1)-1}})^2$ so
\begin{equation} \label{s2}
u_{g(k)}=u_{g(k-1)}(\overline{h_{g(k-1)-1}})^{d_{k-1}}\rhd(\overline{h_{g(k-1)-1}})^{d_{k-1}+2}.
\end{equation}
From (\ref{s1}) and (\ref{s2}) we conclude
\begin{equation*}
y_{k-1}^{-1}(\overline{h_{g(k-1)-1}})^{d_{k-1}} y_k\prec (\overline{h_{g(k-1)-1}})^{d_{k-1}+2},\end{equation*}
which implies that $y_{k-1}^{-1}(\overline{h_{g(k-1)-1}})^{d_{k-1}}y_k =w^{d_{k-1}}$ for some $w\sim \overline{h_{g(k-1)-1}}$,\, but this is possible only if $y_{k-1}=y_k$ which is a contradiction. Hence, Fact 2 is proved in this case.\\

{\bf Case (ii).} Suppose that $y_{k-1}=\alpha$ has not appeared before in $\Delta(\bf s)$, hence,
\begin{equation}\label{hgyug}
\overline{h_{g(k-1)-1}}=y_{k-1} u_{g(k-1)}
\end{equation}
Thus Fact 1 is required as follows
$$y_{k-1}^{-1}(\overline{h_{g(k-1)-1}})^{d_{k-1}}=u_{g(k-1)}(y_{k-1} u_{g(k-1)})^{d_{k-1}-1}\lhd u_{g(k-1)}(\overline{h_{g(k-1)-1}})^{d_{k-1}}x_1^{-1}=u_{g(k)}x_1^{-1}.$$
In order to prove Fact 2, suppose by contrary that
\begin{equation}\label{s1.2}
y_{k-1}^{-1}(\overline{h_{g(k-1)-1}})^{d_{k-1}} y_k \prec u_{g(k)}\end{equation}
On the other hand, by (\ref{hgyug}) and $u_{g(k)}=u_{g(k-1)}(\overline{h_{g(k-1)-1}})^{d_{k-1}}$, we obtain
\begin{equation} \label{sss} u_{g(k)}=u_{g(k-1)} (y_{k-1}  u_{g(k-1)})^{d_{k-1}}\prec (y_{k-1}  u_{g(k-1)})^{d_{k-1}+1}=(\overline{h_{g(k-1)-1}})^{d_{k-1}+1}\end{equation}
From (\ref{s1.2}) and (\ref{sss}) we provide
$$ y_{k-1}^{-1}(\overline{h_{g(k-1)-1}})^{d_{k-1}} y_k \prec (\overline{h_{g(k-1)-1}})^{d_{k-1}+1}$$
which implies that $y_{k-1}^{-1}(\overline{h_{g(k-1)-1}})^{d_{k-1}}y_k=w^{d_{k-1}}$ for some $w\sim \overline{h_{g(k-1)-1}}$, but this is possible only if $y_{k-1}=y_k$, which is a contradiction. This ends the proof.
}
\end{proof}

\section{c-factorization}

\begin{thm} \label{cFactoEpi} Let ${\bf s}$ be an episturmian word with directive word
$\Delta({\bf s})=  x_1 x_2 x_3 \ldots= {y_1}^{d_1} {y_2}^{d_2} {y_3}^{d_3} \ldots$\, , where $x_i,y_i\in {\mathcal A}$ and $y_i\neq y_{i+1}$, for
all $i\geq 1$. If $c(s)=(c_1, c_2, \ldots)$, then
there exist integers $i$ and $j$ such that $c_1\cdots c_k=u_{g(k-j+i+1)}$ for any $k\geq j$. Consequently, we obtain
$c_k=(\overline{h_{g(k-j+i)-1}})^{d_{k-j+i}}$, for all $k\geq j$.
\end{thm}

\begin{proof}{
Let $i= \min \{t: \{y_1, y_2, \ldots, y_t\} = \{1, 2, \ldots, k\}\}$ and $y_i=\alpha$.
Since $y_{i}=\alpha$ has no occurrence in $u_{g(i)}$, we have $u_{g(i)+1}=u_{g(i)}\alpha u_{g(i)}$, hence there exists
an integer $j\geq 3$ satisfying $c_1\cdots c_{j-2}=u_{g(i)}$ and $c_{j-1}=y_i$. Moreover, by Lemma \ref{Pn}(i), we have
\begin{align}
\overline{h_{g(i)-1}}&=y_{i} u_{g(i)}\\
u_{g(i+1)}&=u_{g(i)}(y_{i} u_{g(i)})^{d_i}
\end{align}
Now we are going to prove that $c_j=y_i^{-1}(\overline{h_{g(i)-1}})^{d_i}$. Denote the right side by $w$ and note that
$$c_1\cdots c_{j-1}w=u_{g(i)}\, y_i \, y_i^{-1}(\overline{h_{g(i)-1}})^{d_i}=u_{g(i+1)}$$
It is clear that $w=u_{g(i)}(y_i u_{g(i)})^{d_i-1}$ has at least two occurrences in $c_1\cdots c_{j-1}w=u_{g(i)}(y_{i} u_{g(i)})^{d_i}$. Thus it is enough to prove that $w y_{i+1}\nprec u_{g(i+1)}$. By contrary, suppose that $w y_{i+1}\prec u_{g(i+1)}$ so
\begin{equation}\label{wyu} u_{g(i)}(y_i u_{g(i)})^{d_{i}-1} y_{i+1}\prec u_{g(i)} (y_i u_{g(i)})^{d_i}
\end{equation}
Since $y_i\notin Alph(u_{g(i)})$, (\ref{wyu}) can happen only if $y_{i+1}=y_i$ which is a contradiction. Thus $c_j=w$ as required.\\

Now we claim that the following equation
\begin{equation}\label{cProd}
c_1 c_2\cdots c_k=u_{g(k-j+i+1)}
\end{equation}
holds for any integer $k\geq j$. The statement is true
for $\ell=j$ by the above arguments. We proceed by induction on $k$.
Suppose that $k>j$ and that $c_1 c_2\cdots c_\ell=u_{g(\ell-j+i+1)}$ holds for any integer $\ell$ with $j \leq \ell <k$.
By Lemma \ref{deltaY} (i), it is enough to show that $c_k = (\overline{h_{g(k-j+i)-1}})^{d_{k-j+i}}$.
For this, the two following facts should be proved\\
Fact 1.  $(\overline{h_{g(k-j+i)-1}})^{d_{k-j+i}} \prec u_{g(k-j+i+1)} x_1^{-1}$  \\
Fact 2.   $(\overline{h_{g(k-j+i)-1}})^{d_{k-j+i}}y_{k-j+i+1} \nprec u_{g(k-j+i+1)} $  \\

By Lemma \ref{hu} (i), $\overline{h_{g(k-j+i)-1}} \rhd u_{g(k-j+i)}$, we provide
\begin{equation*}
(\overline{h_{g(k-j+i)-1}})^{d_{k-j+i}} \rhd u_{g(k-j+i)} (\overline{h_{g(k-j+i)-1}})^{d_{k-j+i}-1},
\end{equation*}
which together with $u_{g(k-j+i)} (\overline{h_{g(k-j+i)-1}})^{d_{k-j+i}-1} \prec u_{g(k-j+i+1)} x_1^{-1}$ proves
 Fact 1.\\

 To prove Fact 2, suppose by contrary that $(\overline{h_{g(k-j+i)-1}})^{d_{k-j+i}} y_{k-j+i+1} \prec u_{g(k-j+i+1)}$. By using Lemma \ref{deltaY} (iv), this concludes that

$$(\overline{h_{g(k-j+i)-1}})^{d_{k-j+i}} y_{k-j+i+1} \prec (\overline{h_{g(k-j+i)-1}})^{d_{k-j+i}+2}$$

Since $h_t$ is primitive, it has just $d_{k-j+i}+2$ occurrences in the right side; thus the last relation implies that
$y_{k-j+i+1}$ equals the first letter of $\overline{h_{g(k-j+i)-1}}$, i.e. $y_{k-j+i+1}= y_{k-j+i}$ which is a contradiction.
}\end{proof}

{\bf Remark 1.} By slight modification of the argument used in the proof of Theorem \ref{cFactoEpi}, we find that c-factorization of a standard episturmian word is as follows: $c_1=y_1$ and
$$c_2=\left\{
  \begin{array}{ll}
    {y_2} & \hbox{if $d_1=1$,} \\
    {y_1^{d_1-1}} & \hbox{otherwise.}
  \end{array}
\right.$$
For any integer $m\geq 2$, there exists an integer $n$ such that $c_1 c_2\cdots c_m=u_{g(n)}\alpha_n$, where either $\alpha_n=\epsilon$ or $\alpha_n=y_n$. In addition, the next factor, $c_{m+1}$, is given by
$$c_{m+1}=\left\{
  \begin{array}{ll}
    {y_n} & \hbox{if $\alpha_n=\epsilon$ and $y_n\not \in \{y_1,\cdots,y_{n-1}\}$,} \\
    {(\overline{h_{g(n)-1}})^{d_n}} & \hbox{if $\alpha_n=\epsilon$ and $y_n\in \{y_1,\cdots,y_{n-1}\}$,} \\
    {y_n^{-1}}(\overline{h_{g(n)-1}})^{d_n} & \hbox{otherwise, i.e. if $\alpha_n=y_n$.}
  \end{array}
\right.$$

It is concluded that if $\alpha_n=\epsilon$ and $y_n\not \in \{y_1,\cdots,y_{n-1}\}$, then $c_1\cdots c_{m+1}=u_{g(n)}y_n$; otherwise $c_1\cdots c_{m+1}=u_{g(n+1)}$. Moreover, setting $k_0=Alph({\bf s})$, it is provided that the values $i$ and $j$ in Theorem \ref{cFactoEpi}, satisfy the following equation.
$$j-i=\left\{
  \begin{array}{ll}
    {k_0-1} & \hbox{if $d_1=1$,} \\
    {k_0} & \hbox{otherwise.}
  \end{array}
\right.$$

{\bf Remark 2.} Considering Theorem \ref{zFactoEpi} , Theorem \ref{cFactoEpi} and Remark 1, we conclude that from a point on, the formula $z_k=y_{k-1}^{-1} c_{k+k_0-1-m} y_k$ holds, where $k_0=|Alph({\bf s})|$ and

$$m=\left\{
  \begin{array}{ll}
    {1} & \hbox{if $d_1=1$,} \\
    {0} & \hbox{otherwise.}
  \end{array}
\right.$$

{\bf Remark 3.} From Theorem \ref{cFactoEpi}, by using Remark 1, we obtain that from a point on, $c_k=(\overline{h_{g(k-k_0+m)-1}})^{d_{k-k_0+m}}$, where $k_0$ and $m$ are defined as above.
Now if ${\bf s}$ is standard Sturmian, by using definition and notation of Chapter 2 of \cite{Lot2} about standard words and Sturmian words, it is easily proved that $h_{g(p)-1}=s_{p-1}$, for any integer $p\geq 1$. So in this case, by replacing $k_0=2$, we conclude  that from a point on, $c_k=(\overline{s_{k+m-3}})^{d_{k+m-2}}$. Thus in case $d_1>1$ (resp. $d_1=1$) by calculating the first four factors (resp. first three factors), we conclude Theorem 1 of \cite{Bersavelli} about c-factorization of standard Sturmian words.

\end{document}